\documentclass[journal,onecolumn,12pt]{IEEEtran}
\usepackage{pslatex}
\usepackage{amsfonts,color,morefloats}
\usepackage{amssymb,amsmath,latexsym,amsthm}
\usepackage{cases}

\newcommand{\gf}{{\mathbb{F}}}

\newcommand{\ffm}{{\mathcal{M}}}

\newtheorem{theorem}{Theorem}
\newtheorem{lemma}[theorem]{Lemma}

\newtheorem{corollary}[theorem]{Corollary}

\newtheorem{definition}{Definition}
\newtheorem{example}{Example}

\newtheorem{remark}{Remark}

\begin{document}
\title{The Differential Spectrum of the Power  \\Mapping $x^{p^n-3}$}
\date{\today}
\author{Haode Yan, Yongbo Xia, Chunlei Li, Tor Helleseth, Maosheng Xiong and Jinquan Luo
\thanks{
H. Yan was partially supported by the National Natural
Science Foundation of China (NSFC) under Grant 11801468. Y. Xia was supported in part by NSFC under Grant 61771021, and in part by the Fundamental Research Funds for the Central Universities, South-Central University for Nationalities under Grant CZT20023. C. Li and T. Helleseth were supported by the Research Council of Norway under Grants 247742 and 311646.   M. Xiong was supported by RGC, Hong Kong, under
Grant N$\_$HKUST619$/$17. J. Luo was supported by NSFC under
Grant 11471008. (\emph{Corresponding author:
Yongbo Xia}).

H. Yan is with the School of Mathematics, Southwest Jiaotong University,
Chengdu 610031, China (e-mail: hdyan@swjtu.edu.cn).

Y. Xia is with the Department of Mathematics and Statistics, South-Central University
			for Nationalities, Wuhan 430074, China, and also with the Hubei Key Laboratory of Intelligent Wireless Communications,
  South-Central University for Nationalities, Wuhan 430074, China (e-mail: xia@mail.scuec.edu.cn).

C. Li and T. Helleseth are with the Department of Informatics, University of Bergen, N-5020 Bergen, Norway (e-mail: chunlei.li@uib.no, tor.helleseth@uib.no).

M. Xiong is with the Department of Mathematics,
The Hong Kong University of Science and Technology, Hong Kong (e-mail:
mamsxiong@ust.hk).

J. Luo is with the Hubei Key Laboratory of Mathematical Sciences,
School of Mathematics and Statistics, Central China Normal University,
Wuhan 430079, China (E-mail: luojinquan@mail.ccnu.edu.cn).

 }
}
\maketitle

\begin{abstract}
Let $n$ be a positive integer and  $p$ a prime.	The  power mapping  $x^{p^n-3}$ over $\gf_{p^n}$ has desirable differential properties,  and its  differential spectra for $p=2,\,3$ have been determined. In this paper, for any odd prime $p$, by investigating certain quadratic character sums and some equations over $\gf_{p^n}$,  we determine the differential spectrum of $x^{p^n-3}$ with a unified approach. The obtained result shows that for any given odd prime $p$, the differential spectrum can be expressed explicitly in terms of $n$. Compared with previous results,  a special elliptic curve over $\gf_{p}$ plays an important role in our computation for the general case $p \ge 5$.  \end{abstract}


\section{Introduction}
Let $\gf_{p^n}$ be the finite field with $p^n$ elements and $\gf_{p^n}^*=\gf_{p^n}\setminus \{0\}$, where $p$ is a prime number and $n$ is a positive integer.
Let $F(x)$ be a function from  $\gf_{p^n}$ to itself. The \textit{derivative function}, denoted by $\mathbb{D}_aF$, of $F(x)$ at an element $a$ in $\gf_{p^n}$ is given  by
\[\mathbb{D}_aF(x)=F(x+a)-F(x).\]
For any $a,\,b \in \gf_{p^n}$,  let  \[\delta_F(a,b)=|\{x \in \gf_{p^n}| ~\mathbb{D}_{a}F(x)=b\}|,\]
where $|S|$ denotes the cardinality of a set $S$,
and define
\[\delta(F)=\max\limits_{a \in \gf_{p^n}^*}\max\limits_{b \in \gf_{p^n}}  \delta_F(a,b).\]
A function $F$ is said to be \emph{differentially $\delta$-uniform} iff $\delta(F)=\delta$, and $\delta$ is called the \emph{differential uniformity} of $F(x)$ accordingly \cite{Nyberg93}.
The \emph{differential spectrum} of $F(x)$  is defined as the multiset
\begin{equation*}
\{\,
\delta_F(a,b)\,:\, a \in \gf_{p^n}^*,\, b\in \gf_{p^n}
\,\}.
\end{equation*}
When $F(x)$ is  a power mapping, i.e.,  $F(x)=x^d$ for {a positive integer} $d$,  one easily sees that $\delta_F(a,b)=\delta_F(1,{b/{a^d}})$  for all $a\in \gf_{p^n}^*$ and $b\in \gf_{p^n}$.
That is to say,  the differential spectrum of $F(x)$ is completely determined by the values of $\delta_F(1,b)$ as $b$ runs through $\gf_{p^n}$.
Therefore,  the differential spectrum of a power mapping  can be simplified  as follows.

\begin{definition}\label{def1}
Assume that a power function $F(x)=x^d$ over $\gf_{p^n}$ has differential uniformity $\delta$  and denote
\[\omega_i=|\left\{b\in \gf_{p^n}\mid \delta_F(1, b)=i\right\}|,\,\,0\leq i\leq \delta.\]
The differential spectrum of $F$ is simply defined to be
an ordered sequence
\[
%\mathbb{S}=\left\{(i: \omega_i)\,| \, 0\leq i \leq \delta\right\}.
\mathbb{S} = [\omega_0, \omega_1, \ldots, \omega_{\delta}].
\]
\end{definition}
\smallskip
Due to the differential cryptanalysis \cite{BiSh},
the differential property is one of the most fundamental parameters of cryptographic primitives in block ciphers.
Consequently, it is highly desirable that nonlinear functions for cryptographic applications have low differential uniformity. For example, the AES (Advanced Encryption Standard) uses the inverse function $x \mapsto x^{-1}$ over $\gf_{2^n}$,
which has differential uniformity $4$ for even $n$ and $2$ for odd $n$.
Besides the differential uniformity, the differential spectrum of a nonlinear function also reflects its differential property.
It is usually taken into consideration when one assesses the resistance of a function against differential cryptanalysis and its variants \cite{BCC,BCC2,BP}. Moreover, the differential spectrum of a nonlinear function is also related to the nonlinearity of the function, which is
an important parameter of a function
with respect to linear cryptanalysis \cite{CarletIT18, Charpinpeng19,Matsui}.
\smallskip

In addition to its importance in cryptography, the differential spectrum of a nonlinear function also plays a significant role in sequences, coding theory and combinatorial design. In sequences, the differential spectrum of a power mapping can be used to determine the cross-correlation between $m$-sequences and their decimation sequences \cite{Dobbertin2001};
in coding theory, the differential spectrum is highly related to the number of low weight codewords in some linear codes \cite{BCC, Carlet, C-PAMC}; and in combinatorial designs, some new $2$-designs can be constructed from  differentially two-valued functions \cite{tangding}.
Therefore,  it is an interesting topic to completely determine the differential spectrum of a nonlinear function with low differential uniformity.
This problem is, nevertheless, relatively challenging.
So far, only a few infinite families of power mappings have known differential spectra, which are listed in Table \ref{table-1}.

\smallskip

The investigation of differential spectra of power mappings over finite fields, to the best of our knowledge, first appeared in \cite{Dobbertin2001}, where the authors
considered the differential spectrum of $x^d$ over $\gf_{3^n}$ with odd $n$ and $d=2\cdot 3^{\frac{n-1}{2}}+1$ (known as the ternary Welch exponent). The result obtained there was then used to resolve the ternary Welch conjecture that the cross-correlation function between an $m$-sequence of period $3^n-1$ and its ternary Welch-decimated sequence takes exactly three values.
Blondeau, Canteaut and Charpin later in \cite{BCC} dedicated their research focus to the differential spectra of several power mappings in the binary case, including quadratic power mappings, Bracken-Leander power mapping and Kasami power mapping, and they proposed some conjectures.
The differential properties of the power mappings $x^{2^{t}-1}$ over $\gf_{2^n}$ were later investigated in \cite{BCC2} and \cite{BP}, where
the differential spectra of  $x^{2^{t}-1}$ for certain special $t$'s were determined.
Xiong et al. in \cite{XYY} proved one of the conjectures in \cite{BCC} about the differential spectra of the power functions with Niho exponents.
Very recently, for the power mapping $x^{2^{3k}+2^{2k}+2^k-1}$ over $\gf_{2^n}$ with $n=4k$,  Li et al. \cite{Li} determined its differential spectrum, which gives an affirmative answer to the conjecture proposed in \cite{Lilya}.
In recent years some research progress has also been made for the nonbinary cases.
Choi et. al  \cite{CHNC} computed  the differential spectra of two  power functions $x^{\frac{p^k+1}{2}}$ and $x^{\frac{p^n+1}{p^m+1}+\frac{p^n-1}{2}}$, where the conditions on $p, n, k, m$ are listed in Table \ref{table-1}.
The differential spectra of the family of $p$-ary Kasami power permutation $x^{p^{2k}-p^k+1}$ over $\gf_{p^n}$ with $\gcd(n,k)=1$ and its generalized family with $\gcd(n,k)=e$ were investigated in \cite{Yan} and \cite{Lei2020}, respectively.

\smallskip

Our study in this paper originates from  the work of Helleseth,  Rong and Sandberg \cite{HRS}, where they intensively studied the differential properties of a number of  power functions and presented several families of  APN functions.
In particular, the differential properties of the power function $x^{p^n-3}$ were characterized as follows.

   \begin{theorem}\cite[Theorem 7]{HRS}\label{hellresult} Let $d=p^n-3$   and let $F(x)=x^d$  be a mapping over    $\gf_{p^n}$.

 \noindent  (\textrm{i}) If $p=2$, then $\delta(F)=2$    when $n$ is odd and $\delta(F)=4$  when $n$ is even.

    \noindent  (\textrm{ii})  If $p$ is an odd prime, then $1\leq \delta(F)\leq 5$.

   \noindent     (\textrm{iii})  If $n>1$ is odd  and $p=3$,  then $\delta(F)=2$.
   \end{theorem}

Given Theorem \ref{hellresult}, a natural question arises: what is the differential spectrum of the power mapping $x^{p^n-3}$ over $\gf_{p^n}$? There are some partial answers to this question. By setting $\frac{1}{0}=0$,  the above power mapping can be rewritten as  $F(x)=x^{-2}$.  When $p=2$, it is equivalent to the inverse function $x^{-1}$ over  $\gf_{2^n}$, of which the differential spectrum has been determined in \cite{BCC}. Recently, for $p=3$  the differential spectrum of $F(x)=x^{p^n-3}$ was completely determined in \cite{XZLH}, where the authors characterized the conditions on $b$ such that the derivative equation $\mathbb{D}_1F(x)=F(x+1)-F(x)=b$ has two and four roots in $\mathbb{F}_{3^n}$, respectively. The method used in \cite{XZLH} relies heavily on the characteristic $p=3$,  and  it is not clear how it may work for the general prime $p$.

\smallskip

   In this paper, for any odd prime $p$, we present a unified approach to studying the differential spectrum of  $x^{p^n-3}$, which is different from that used in \cite{XZLH}. In our approach, we investigate several related equations in details and
establish a connection between the differential spectrum of $x^{p^n-3}$ and the quadratic character sums with two quartic polynomials.  For the case $p=3$, the two quartic polynomials are essentially quadratic ones and hence the two quadratic character sums  can be evaluated directly; when $p \ge 5$, both of the quadratic character sums are related to a single elliptic curve over $\gf_{p}$, and they can be computed by the theory of elliptic curves. As a result, for any given odd prime $p$, the differential spectrum of $x^{p^n-3}$ can be derived and be expressed explicitly in terms of $n$.  Therefore, our work completely settles the unsolved problem about the differential spectrum of  $x^{p^n-3}$ in Theorem 7 of \cite{HRS}.

  \smallskip
The rest of this paper is organized as follows. Section \ref{pre} introduces some quadratic character sums and the related theory of elliptic curves over $\gf_{p^n}$. In Section \ref{sec3}, we will determine the number of solutions to an equation system, which is dependent on a quadratic character sum presented in Section \ref{pre}.  With the preparations in Sections \ref{pre} and \ref{sec3}, the differential spectrum of $x^{p^n-3}$ is  computed in Section \ref{sec4}. Section \ref{sec5} concludes this paper.
\begin{table}[t]\label{table-1}
\caption{Some power functions $F(x)=x^d$ over $\gf_{p^n}$ with known  differential spectrum}
\centering
\begin{tabular}{|c||c|c|c|c|}
\hline
%inserts double horizontal lines
$p$&$d$ & Condition & $\delta(F)$ & Ref. \\
[0.5ex]
\hline
% inserts single horizontal line
2& $2^t+1$& $\gcd(t,n)=s$ & $2^s$ &\cite{BCC}\\
\hline
2& $2^{2t}-2^t+1$& $\gcd(t,n)=s$, $n/s$ odd &$2^s$ &\cite{BCC}\\
\hline
2&$2^n-2$ &$n\geq 2$ & $2$ or $4$ &\cite{BCC}\\
\hline
%2&$2^{2k}+2^k+1$ &$n=4k$, $k$ odd & 4 &\cite{BCC}\\
%\hline
2&$2^{2k}+2^k+1$ &$n=4k$ & 4 &\cite{BCC}, \cite{XY}\\
\hline
2&$2^t-1$ &$t=3,n-2$ & 6 or 8  &\cite{BCC2}\\
\hline
2&$2^t-1$ &$t=n/2, n/2+1$, $n$ even & $2^{n/2}-2$ or $2^{n/2}$  &\cite{BCC2}\\
\hline
2&$2^t-1$ &$t=(n-1)/2,(n+3)/2$, $n$ odd & $6$ or $8$ &\cite{BP}\\
\hline
2&$2^m+2^{(m+1)/2}+1$ &$n=2m$, $m\geq5$ odd &$8$ &\cite{XYY}\\
\hline
2&$2^{m+1}+3$ &$n=2m$, $m\geq5$ odd &$8$ &\cite{XYY}\\
\hline
2&$2^{3k}+2^{2k}+2^{k}-1$ & $n=4k$ & $2^{2k}$ & \cite{Li} \\
\hline
  $3$&$2\cdot 3^{(n-1)/2}+1$&$n$ odd &$4$&\cite{Dobbertin2001}\\
\hline
$3$&$3^n-3$ &$n$ odd, $n\equiv 2(\mathrm{mod}~4)$, or $n\equiv 0(\mathrm{mod}~4)$ &$2,4,$ or $5$ &\cite{XZLH}\\
\hline
$p$ odd &$(p^k+1)/2$&$e=\gcd(n,k)$&$(p^e-1)/2$ or $p^e+1$&\cite{CHNC}\\
\hline
$p$ odd&$(p^n+1)/(p^m+1)+(p^n-1)/2$&$p \equiv 3 ~(\mathrm{mod}~4)$, $n$ odd, $m|n$&$(p^m+1)/2$&\cite{CHNC}\\
\hline
%$p$ odd&$p^{2k}-p^k+1$ &$n$ odd, $\gcd(n,k)=1$ &$p+1$ &\cite{Yan}\\
%\hline
$p$ odd&$p^{2k}-p^k+1$ &$\gcd(n,k)=e$, $n/e$ odd,  &$p^e+1$ &\cite{Yan}, \cite{Lei2020}\\
\hline

\end{tabular}
%\label{table:1}
\end{table}

\section{Some quadratic sums and the theory of elliptic curves}\label{pre}
From now on, we always assume that $p$ is an odd prime and $\eta$ is the quadratic multiplicative character of $\gf_{p^n}^*$.  It is convenient to extend the definition of $\eta$ to $\gf_{p^n}$ by setting $\eta(0)=0$. For an element $\beta\in \gf_{p^n}$, if $\eta(\beta)=1$, then it has exactly two square roots in $\gf_{p^n}$, which are denoted by $\pm\sqrt{\beta}$ throughout this paper. In the sequel, for convenience we also frequently adopt the convention that $\frac{1}{0}:=0$.

\smallskip

Let $\gf_{p^n}[x]$ denote the polynomial
ring over $\gf_{p^n}$. We shall consider the sums involving the quadratic character and having polynomial arguments of the form
\begin{eqnarray*} \label{2:cha-f} \sum_{x\in\gf_{p^n}}\eta(f(x))\end{eqnarray*}
 with $f(x)\in\gf_{p^n}[x]$.
It is clear that the case of linear $f(x)$ is trivial. When $f(x)$ is quadratic, the explicit formula was given in \cite{FF}.
\begin{lemma} \cite[Theorem 5.48]{FF}\label{charactersumquadratic} Let $f(x)=a_2x^2+a_1x+a_0\in\gf_{p^n}[x]$ with $p$ odd and $a_2\neq0$. Put $d=a^2_1-4a_0a_2$ and let $\eta$ be the quadratic character of $\gf_{p^n}$. Then
\begin{eqnarray*}
\sum_{x\in\gf_{p^n}}\eta(f(x))=
\begin{cases}
	-\eta(a_2), & \text{ if } d\neq 0, \\
	(p^n-1)\eta(a_2), & \text{ if } d = 0.
\end{cases}
\end{eqnarray*}
\end{lemma}

\smallskip

As it will be seen in Sections \ref{sec3} and \ref{sec4}, the computation of the differential spectrum of the power mapping $x^{p^n-3}$ over $\gf_{p^n}$ boils down to evaluating two specific character sums
\begin{eqnarray}
\label{2:f1} \lambda_{1,p^n}:=\sum_{x\in\gf_{p^n}}\eta\left((x^2-4)(-3x^2-4)\right),\end{eqnarray}
and
\begin{eqnarray}
\label{2:f2} \lambda_{2,p^n}:=\sum\limits_{x\in\gf_{p^n}}\eta\left((x^2+1)(x^2+4x+1)\right).
\end{eqnarray}
Note that in the case of $p=3$ the above character sums can be easily computed. To be more concrete,
one has $-3x^2-4=-4$ and $x^2+4x+1=(x+2)^2$, then the polynomials involved in $\lambda_{1,3^n}$ and $\lambda_{2,3^n}$ are essentially quadratic ones. Hence Lemma \ref{charactersumquadratic} can be applied directly and we have
\begin{eqnarray} \label{pre:p=3}
\lambda_{1,3^n}=-\eta(-1)\,\,\,\,\mbox{and}\,\,\,\, \lambda_{2,3^n}=-1-\eta(2).
\end{eqnarray}

\smallskip

When $p \ge 5$, the situation is quite different. The polynomials involved in $\lambda_{1,p^n}$ and $\lambda_{2,p^n}$ are of degree 4, these character sums correspond to the elliptic curves $y^2=(x^2-4)(-3x^2-4)$ and $y^2=(x^2+1)(x^2+4x+1)$ over  $\gf_p$ respectively. Generally speaking, by the theory of elliptic curves in \cite{SilveEC}, there is no explicit formula for the evaluation of such character sums in general, except for some very special kinds of elliptic curves that are very rare. The following theorem provides an efficient  method to evaluate $\lambda_{1,p^n}$ and $\lambda_{2,p^n}$ for $p \ge 5$ based on the theory of elliptic curves.

\begin{theorem}\label{thm gama} Let $p\geq 5$. Denote by $N_p$ the number of $(x,y) \in \gf_{p}^2$ satisfying the equation
\begin{equation}\label{ec}
E: y^2=x(x-1)(x+3).
\end{equation}
Define $a=N_{p}-p$ and let $\gamma_{p,1}$ and $\gamma_{p,2}$ be the two roots of the quadratic polynomial $T^2+aT+p$ in the complex number field.
Define \begin{equation}\label{keysums} \Gamma_{p,\,n}:=\sum\limits_{x\in\gf_{p^n}}\eta(x(x-1)(x+3)).\end{equation}
Then
\begin{eqnarray} \label{2:eva-cha}
\Gamma_{p,\,n}=-\gamma_{p,1}^n-\gamma_{p,2}^n,\,\,\,\, \lambda_{1,p^n}=\Gamma_{p,\,n}-\eta(-3)\,\,\,\,\mbox{and}\,\,\,\,\lambda_{2,p^n}=\Gamma_{p,\,n}-1.
\end{eqnarray}
\end{theorem}
\begin{proof}
The equation (\ref{ec}) defines an elliptic curve $E$ over $\gf_{p}$. The quadratic character sum $\Gamma_{p,n}$ defined in (\ref{keysums}) is closely related to the number of $\gf_{p^n}$-rational points (with the extra point at infinity) on $E$, which is actually equal to
$p^n+1+\Gamma_{p,\,n}$. By the theory of elliptic curves (see \cite[Theorem 2.3.1, Chap. V]{SilveEC}), we have \begin{equation*}\Gamma_{p,n}=-\gamma_{p,1}^n-\gamma_{p,2}^n. \end{equation*}
The Weil bound for $\Gamma_{p,n}$ is that $|\Gamma_{p,\,n}|\leq 2\sqrt{p^n}$ (see \cite[Corollary 1.4, Chap. V]{SilveEC}). Note that $a=\Gamma_{p,1}$, which is an integer. Thus, we have $a^2< 4p$ and $\gamma_{p,1}\neq \gamma_{p,2}$.

Now using $\Gamma_{p,n}$ we can evaluate $\lambda_{1,p^n}$ as follows:
\begin{equation*}
\begin{array}{lcl}
\sum\limits_{x\in\gf_{p^n}}\eta\left((x^2-4)(-3x^2-4)\right)
&=&1+2\sum\limits_{\eta(u)=1}\eta\left((u-4)(-3u-4)\right)\\
&=&1+2\sum\limits_{\eta(u)=1}\eta\left((1-\frac{4}{u})(-3-\frac{4}{u})\right)\\
&=&1+2\sum\limits_{\eta(u)=1}\eta\left((1-u)(-3-u)\right)\\
&=&1+\sum\limits_{u\in\gf_{p^n}}\left(1+\eta(u)\right)\eta\left((u-1)(u+3)\right)-\eta(-3)\\
&=&\sum\limits_{u\in\gf_{p^n}}\eta\left((u-1)(u+3)\right)+\Gamma_{p,\,n}+1-\eta(-3).\\
\end{array}
\end{equation*}
The first term $\sum\limits_{u\in\gf_{p^n}}\eta\left((u-1)(u+3)\right)=-1$ according to Lemma \ref{charactersumquadratic}. Thus we have the desired result for $\lambda_{1,p^n}$.

As for $\lambda_{2,p^n}$, let $\frac{x^2+4x+1}{x^2+1}=u$. Then,  $u$ and $x$ satisfy
\begin{equation}\label{ueqn1}(u-1)x^2-4x+(u-1)=0.
\end{equation}
It is easy to see that $x=0$ if and only if $u=1$.  When $u\neq1$, (\ref{ueqn1}) is a quadratic equation in the variable $x$, and it has solutions in $\gf_{p^n}$ if and only if  $\eta(\Delta)=\eta((u+1)(-u+3))=1$ or $0$.  If $u=-1$ (resp.  $u=3$), then $x=-1$ (resp. $x=1$) is the unique solution of (\ref{ueqn1}). If $u\neq1$ and $\eta((u+1)(-u+3))=1$,  there are two distinct $x$'s satisfying  (\ref{ueqn1}). Thus we have
\begin{equation}\label{ueqn2}
\sum_{x\in\gf_{p^n}}\eta\left(\frac{x^2+4x+1}{x^2+1}\right)=\eta(1)+\eta(-1)+\eta(3)+2\sum_{  u\neq1,\eta((u+1)(-u+3))=1}\eta(u),
\end{equation}
where we may adopt the convention that $\frac{1}{0}:=0$ and $\eta(0)=0$.
Furthermore,  \begin{equation*}\begin{array}{lcl}
&&2 \sum\limits_{u\neq1,\eta((u+1)(-u+3))=1}\eta(u)\\
&=&\sum\limits_{u\neq1}(1+\eta((u+1)(-u+3)))\eta(u)-\eta(-1)-\eta(3)\\
&=&\sum\limits_{u\in\gf_{p^n}}(1+\eta((u+1)(-u+3)))\eta(u)-2\eta(1)-\eta(-1)-\eta(3)\\
&=&\sum\limits_{u\in\gf_{p^n}}\eta(u)+\sum\limits_{u\in\gf_{p^n}}\eta(u(u+1)(-u+3))-2\eta(1)-\eta(-1)-\eta(3)\\
&=&\sum\limits_{u\in\gf_{p^n}}\eta((-u)(-u+1)(u+3))-2\eta(1)-\eta(-1)-\eta(3)\\
&=&\Gamma_{p,n}-2\eta(1)-\eta(-1)-\eta(3),
\end{array}
\end{equation*}
where the fourth equality holds since $\sum\limits_{u\in\gf_{p^n}}\eta(u)=0$. This together with (\ref{ueqn2}) yields
\begin{equation*}\sum_{x\in\gf_{p^n}}\eta\left(\frac{x^2+4x+1}{x^2+1}\right)=\Gamma_{p,\,n}-1.
\end{equation*}
Since \begin{equation*}\lambda_{2,p^n}=\sum_{x\in\gf_{p^n}}\eta\left(\frac{x^2+4x+1}{x^2+1}\right)\eta\left((x^2+1)^2\right)=\sum_{x\in\gf_{p^n}}\eta\left(\frac{x^2+4x+1}{x^2+1}\right),
\end{equation*} it follows the desired evaluation of $\lambda_{2,p^n}$.
\end{proof}

\begin{remark}
We  emphasize that a unified explicit formula of the character sum $\Gamma_{p,n}$ for all primes $p \ge 5$ and positive integers $n$ may not exist at all; and we have the same situation for $\lambda_{1,p^n}$ and $\lambda_{2,p^n}$.
However, Theorem \ref{thm gama} enables us to give a practical and efficient algorithm for evaluating these character sums, which can be described as follows:

\begin{itemize}
\item \textbf{Step 1:} For each given $p \ge 5$, compute the quantity $N_p$, which can be easily computed for most practical values of $p$  by Magma . Then, we get $a=N_p-p$.

 \item \textbf{Step 2:} Determine the two roots $\gamma_{p,1}$ and $\gamma_{p,2}$ of the polynomial $x^2+aT+p$ in the complex number field, which are
$$\gamma_{p,1}=\frac{-a+\sqrt{a^2-4p}}{2}, \quad \gamma_{p,2}=\frac{-a-\sqrt{a^2-4p}}{2}.$$

\item \textbf{Step 3:} Compute $\Gamma_{p,n}$,  $\lambda_{1,p^n}$ and $\lambda_{2,p^n}$ according to (\ref{2:eva-cha}).

\end{itemize}

Note that $a= \Gamma_{p,1}$ and thus in Step 1 we can compute the value of $a$ directly according to (\ref{keysums}). Utilizing the above algorithm, one knows that for any given prime $p \ge 5$, the character sums $\Gamma_{p,n}$, $\lambda_{1,p^n}$
 and $\lambda_{2,p^n}$ can be computed and expressed explicitly in terms of $n$. The following example illustrates the above procedure of calculating them.
\end{remark}

\begin{example} \label{examp1}For $p=5$, by using Magma we can obtain $N_5=7$, hence $a=2$. So we have $\gamma_{5,1},\gamma_{5,2}=-1 \pm 2 \sqrt{-1}$, hence
\[\Gamma_{5,n}=-\left(-1 + 2 \sqrt{-1}\right)^n-\left(-1 - 2 \sqrt{-1}\right)^n. \]
For $p=7$, by using Magma we can obtain $a=0$ by (\ref{keysums}). So we have $\gamma_{7,1},\gamma_{7,2}=\pm  \sqrt{-7}$. Then, we get
\[\Gamma_{7,n}=-\left(1+(-1)^n\right) \sqrt{-7}^{\,n}. \]
The values of $\Gamma_{p,n}$ for other $p$ can be obtained similarly. Once the value of $\Gamma_{p,n}$ is obtained, so are the values of $\lambda_{1,p^n}$ and $\lambda_{2,p^n}$.
\end{example}

In Table \ref{valp}, for all primes $p \le 1000$ we list the values of $a$ ($\Gamma_{p,1}$) computed with Magma.

\begin{table}
\caption{The values of  $\Gamma_{p,1}$ for $p\leq 1000$}\label{valp}
\centering
\begin{tabular}{|c||c|c|c|c|c|c|c|c|c|c|c|c|c|c|}
\hline
%inserts double horizontal lines
$p$&$5$ & $7$ & $11$ & $13$ & $17$&$19$ & $23$ & $29$ & $31$ & $37$&$41$&$43$&$47$&$53$\\

\hline
$\Gamma_{p,1}$&$2$ & $0$ & $-4$ & $2$ & $-2$&$4$ & $8$ & $-6$ & $-8$ & $-6$& $6$& $-4$& $0$& $2$\\
\hline
\hline
$p$&$59$ & $61$ & $67$ & $71$ & $73$&$79$ & $83$ & $89$ & $97$ & $101$&$103$&$107$&$109$&$113$\\

\hline
$\Gamma_{p,1}$&$-4$ & $2$ & $4$ & $-8$ & $-10$&$8$ & $4$ & $6$ & $-2$ & $18$& $-16$& $12$& $2$& $-18$\\
\hline
\hline
$p$&$127$ & $131$ & $137$ & $139$ & $149$&$151$ & $157$ & $163$ & $167$ & $173$&$179$&$181$&$191$&$193$\\

\hline
$\Gamma_{p,1}$&$8$ & $4$ & $6$ & $12$ & $-14$&$16$ & $2$ & $-12$ & $-24$ & $-6$& $-12$& $-6$& $0$& $-2$\\
\hline
\hline
$p$&$197$ & $199$ & $211$ & $223$ & $227$&$229$ & $233$ & $239$ & $241$ & $251$&$257$&$263$&$269$&$271$\\

\hline
$\Gamma_{p,1}$&$18$ & $-16$ & $20$ & $8$ & $-12$&$-22$ & $-10$ & $16$ & $-18$ & $-20$& $-2$& $8$& $10$& $-8$\\
\hline
\hline

$p$&$277$ & $281$ & $283$ & $293$ & $307$&$311$ & $313$ & $317$ & $331$ & $337$&$347$&$349$&$353$&$359$\\

\hline
$\Gamma_{p,1}$&$26$ & $-26$ & $28$ & $18$ & $-12$&$24$ & $6$ & $-6$ & $-20$ & $-18$& $12$& $-30$& $-2$& $24$\\
\hline
\hline
$p$&$367$ & $373$ & $379$ & $383$ & $389$&$397$ & $401$ & $409$ & $419$ & $421$&$431$&$433$&$439$&$443$\\

\hline
$\Gamma_{p,1}$&$8$ & $10$ & $-20$ & $0$ & $2$&$-14$ & $30$ & $6$ & $-12$ & $10$& $-32$& $14$& $0$& $-20$\\
\hline
\hline

$p$&$449$ & $457$ & $461$ & $463$ & $467$&$479$ & $487$ & $491$ & $499$ & $503$&$509$&$521$&$523$&$541$\\

\hline
$\Gamma_{p,1}$&$14$ & $22$ & $26$ & $-8$ & $36$&$16$ & $32$ & $12$ & $-12$ & $-24$& $-6$& $-26$& $-4$& $18$\\
\hline
\hline

$p$&$547$ & $557$ & $563$ & $569$ & $571$&$577$ & $587$ & $593$ & $599$ & $601$&$607$&$613$&$617$&$619$\\

\hline
$\Gamma_{p,1}$&$-44$ & $26$ & $-28$ & $-10$ & $-36$&$-2$ & $44$ & $14$ & $-24$ & $38$& $40$& $-38$& $-42$& $44$\\
\hline
\hline
$p$&$631$ & $641$ & $643$ & $647$ & $653$&$659$ & $661$ & $673$ & $677$ & $683$&$691$&$701$&$709$&$719$\\

\hline
$\Gamma_{p,1}$&$-16$ & $14$ & $-12$ & $-8$ & $-6$&$-12$ & $10$ & $-34$ & $2$ & $-4$& $4$& $-6$& $10$& $32$\\
\hline
\hline
$p$&$727$ & $733$ & $739$ & $743$ & $751$&$757$ & $761$ & $769$ & $773$ & $787$&$797$&$809$&$811$&$821$\\

\hline
$\Gamma_{p,1}$&$-48$ & $-14$ & $4$ & $8$ & $-24$&$-38$ & $22$ & $-2$ & $18$ & $-28$& $-22$& $-26$& $-4$& $-30$\\
\hline
\hline
$p$&$823$ & $827$ & $829$ & $839$ & $853$&$857$ & $859$ & $863$ & $877$ & $881$&$883$&$887$&$907$&$911$\\

\hline
$\Gamma_{p,1}$&$16$ & $28$ & $50$ & $24$ & $10$&$-42$ & $12$ & $32$ & $18$ & $-50$& $4$& $-8$& $-4$& $-16$\\
\hline
\hline

$p$&$919$ & $929$ & $937$ & $941$ & $947$&$953$ & $967$ & $971$ & $977$ & $983$&$991$&$997$& & \\

\hline
$\Gamma_{p,1}$&$-16$ & $-50$ & $-42$ & $-6$ & $-12$&$54$ & $16$ & $-36$ & $30$ & $24$& $-40$& $26$& & \\
\hline

\end{tabular}
%\label{table:1}
\end{table}

\begin{remark}\label{remaka=0}
If $a=0$, then $\gamma_{p,1},\gamma_{p,2}=\pm \sqrt{-p}$, and we have a simple expression of $\Gamma_{p,n}$ as
 \begin{eqnarray*}
\Gamma_{p,\,n}=\left\{
\begin{array}{cl}
0&{\rm ~if~}n ~{\rm is~odd},\\
-2\sqrt{-1}^{\,n}p^{n/2}&{\rm ~if~}n ~{\rm is~even}.\\
\end{array} \right.\ \
\end{eqnarray*}
It was known that $a=0$ if and only if the elliptic curve $E$ defined over $\gf_p$ in (\ref{ec}) is supersingular, and there is an explicit and efficient formula to determine whether or not $E$ is supersingular (see \cite[Theorem 4.1, Chap. V]{SilveEC}). In particular, for $p \le 1000$, the elliptic curve $E$ defined over $\gf_{p}$ is supersingular if $p=7,47,191,383$ and $439$, thus in these cases the values $\Gamma_{p,n}$, $\lambda_{1,p^n}$ and $\lambda_{2,p^n}$ can be presented in a more compact form.
\end{remark}

\section{The number of solutions to an equation system} \label{sec3}

Let $d=p^n-3$ with $p$ being an odd prime. Denote by $\ffm$ the set of solutions  $(x_1,x_2,x_3,x_4)\in(\gf_{p^n})^4$ of the equation system
\begin{eqnarray}\label{eq-sys}
\left\{
\begin{array}{lllll}
x_1-x_2+x_3-x_4&=&0,\\
x^{d}_1-x^{d}_2+x^{d}_3-x^{d}_4&=&0,
\end{array} \right.\ \
\end{eqnarray}
and $M=|\ffm|$. In this section we shall compute the value of $M$, which plays an important role in determining the differential spectrum of the power mapping $x^{p^n-3}$ over $\gf_{p^n}$.

To this end, we need to make some preparations. Define
\[\ffm_i=\left\{(x_1,x_2,x_3,x_4) \in \ffm~|~ x_i=0\right\}, \quad i=1,2,3,4,\]
and
\[\ffm^{\circ}=\left\{(x_1,x_2,x_3,x_4) \in \ffm~|~  x_1x_2x_3x_4 \neq 0\right\}.\]
It is trivial to see that
\begin{eqnarray} \label{M:eq1}
|\ffm_i \cap \ffm_j| =\left\{\begin{array}{cl}
p^n, & ~\mathrm{if}~ (i,j) \in \{(1,2),(1,4), (2,3), (3,4)\}, \\
1, & ~\mathrm{if} ~ (i,j) \in \{(1,3), (2,4)\},
\end{array}\right.\end{eqnarray}
and
\begin{eqnarray} \label{M:eq2}
|\ffm_i \cap \ffm_j \cap \ffm_k| =1 \text{ for any } 1 \le i<j<k \le 4, \quad \left|\cap_{i=1}^4 \ffm_i \right|=1.\end{eqnarray}
Next we compute $|\ffm_i|$ ($1 \le i \le 4$) and $|\ffm^\circ|$.

The following result about a quartic equation over $\gf_{p^n}$ is useful for computing $|\ffm_i|$ ($1 \le i \le 4$).
Before we give the result, we recall from Section \ref{pre} that for any $\beta\in\gf_{p^n}$ with $\eta(\beta)=1$, the two square roots of $\beta$ are denoted by $\sqrt{\beta}$ and $-\sqrt{\beta}$.

\begin{lemma}\label{Lem-g1}Let $p\geq3$ be an odd prime, and $g_1(x)=x^4+2x^3+x^2+2x+1\in \gf_{p^n}[x]$. Denote by $T_1$ the number of roots of $g_1(x)$ in $\gf_{p^n}$. Then, we have
\begin{eqnarray*}
T_1=\left\{
\begin{array}{lllll}
0,\,\, \mathrm{if}~\eta(2)=-1,~\mathrm{or}~\eta(2)=\eta(-7)=1~\mathrm{but}~\eta(-1+2\sqrt{2})=-1,\\
1, \,\,\mathrm{if}~p=7~\mathrm{and}~n~\mathrm{is~odd},\\
2,\,\, \mathrm{if}~\eta(2)=1~\mathrm{and}~\eta(-7)=-1,\\
3,\,\, \mathrm{if}~p=7~\mathrm{and}~n~\mathrm{is~even},\\
4,\,\, \mathrm{if}~\eta(2)=\eta(-7)=\eta(-1+2\sqrt{2})=1.
\end{array} \right.\ \
\end{eqnarray*}
\end{lemma}

\begin{proof}
Let $x \in \gf_{p^n}$ be a solution of $g_1(x)$, then we have
\begin{equation}\label{Lem-g1-eq1}
\left(x+\frac{1}{x}\right)^2+2\left(x+\frac{1}{x}\right)-1=0,
\end{equation}
which can be regarded as a quadratic equation in variable $z=x+\frac{1}{x}$ with discriminant $\Delta=2^2-4 \cdot (-1)=8$. If $\eta(\Delta)=-1$, that is, $\eta(2)=-1$, then $T_1=0$. Now suppose $\eta(\Delta)=\eta(2)=1$. Solving (\ref{Lem-g1-eq1}), we have
\[ x+\frac{1}{x}= -1 \mp \sqrt{2},\]
which implies that
\begin{eqnarray} \label{3:eq2} x^2+(1\pm\sqrt2)x+1=0. \end{eqnarray}
To solve (\ref{3:eq2}) over $\gf_{p^n}$, we compute the corresponding discriminants which are $\Delta_{1}=-1+2\sqrt{2}$, $\Delta_{2}=-1-2\sqrt{2}$. Noting that $\Delta_1 \cdot \Delta_2=-7$, there are two cases to consider:

{\it Case 1. $\Delta_1\cdot\Delta_2=0$.} This occurs if and only if $p=7$. In this case, $3^2=2$ hence we may take $\sqrt{2}=3$, then we have $(\Delta_1,\Delta_2)=(5,0)$. For $\Delta_2=0$, the corresponding equation (\ref{3:eq2}) is always solvable with a unique solution. As for $\Delta_1=5$, note that $5$ is a nonsquare in  $\gf_7$. Therefore, if $n$ is odd, then $\eta(5)=-1$ and the equation (\ref{3:eq2}) corresponding to $\Delta_1$ is not solvable in $\gf_{p^n}$, that is, $T_1=1$. On the other hand, if $n$ is even, then $\eta(5)=1$ and the equation (\ref{3:eq2}) corresponding to $\Delta_1$ has two distinct solutions in $\gf_{p^n}$, so in this case we have $T_1=3$.

{\it Case 2. $\Delta_1\cdot\Delta_2 \ne 0$.} Then $p\neq 7$. If $\eta(\Delta_1)=\eta(\Delta_2)=1$, then the equations (\ref{3:eq2}) corresponding to both $\Delta_1$ and $\Delta_2$ are solvable with two distinct solutions, so $T_1=4$. If $\eta(\Delta_1)=\eta(\Delta_2)=-1$, then the equation (\ref{3:eq2}) is not solvable for either $\Delta_1$ or $\Delta_2$, hence $T_1=0$. On the other hand, if $\eta(\Delta_1) \cdot \eta(\Delta_2)=\eta(-7)=-1$, then the corresponding equation (\ref{3:eq2}) is solvable with two distinct solutions in $\gf_{p^n}$ for exactly one of $\Delta_1$ and $\Delta_2$, that is, $T_1=2$.

Summarizing all the above cases we obtain the desired formula for $T_1$. This completes the proof of Lemma \ref{Lem-g1}.
\end{proof}

\begin{remark}\label{remark of T1}
For any given odd prime $p$  and positive integer $n$, in order to get the exact value of $T_1$, one first needs to compute $\eta(2)$ and $\eta(-7)$ in $\gf_{p^n}$,
 which is straightforward according to the Legendre symbols $\left(\frac{2}{p}\right)$, $\left(\frac{-7}{p}\right)$ and the parity of $n$.  If $\eta(2)=\eta(-7)=1$,
then one further needs to check the value of $\eta(-1+2\sqrt{2})$. This can be handled efficiently by the following way:
\begin{itemize}
\item when $\left(\frac{2}{p}\right)=1$, then $-1+2\sqrt{2}$ is an element in $\gf_{p}$, and it is always a square in $\gf_{p^2}$. Thus, the element $-1+2\sqrt{2}$ is a square of $\gf_{p^n}$ iff
$-1+2\sqrt{2}$ is a square of $\gf_{p}$ or $n$ is even;

\item when $\left(\frac{2}{p}\right)=-1$, then $-1+2\sqrt{2}$ is an element in $\gf_{p^2}\setminus \gf_{p}$, and since $\eta(2)=1$, $n$ must be even. Thus,
$\eta(-1+2\sqrt{2})=1$ iff $-1+2\sqrt{2}$ is a square in $\gf_{p^2}$, or $n$ is a multiple of $4$.

\end{itemize}

An alternative approach to computing $\eta(-1+2\sqrt{2})$ is based on investigating the polynomial $(x^2+1)^2-8$, that is, $x^4+2x^2-7$. We have $\eta(-1+2\sqrt{2})=1$
if and only if  $x^4+2x^2-7$ has a root in $\gf_{p^n}$. In order to determine whether the polynomial $x^4+2x^2-7\in\gf_{p}[x]$ has a root in $\gf_{p^n}$, it suffices to verify whether it has roots in $\gf_p$ and $\gf_{p^2}$. Then, combined with the parity of $n$ or $n/2$, we can obtain the desired result. The details are omitted here.

\end{remark}

\begin{lemma}\label{L3-Mi} With the notation introduced above,  for any $1 \le i \le 4$, we have $|\ffm_i|=p^n+(1+T_1)(p^n-1)$, where $T_1$ is given in Lemma \ref{Lem-g1}.
\end{lemma}
\begin{proof} It is easy to see that $|\ffm_i|=|\ffm_4|$ for any $1 \le i \le 4$. So we only consider $\ffm_4$, that is, $x_4=0$ in (\ref{eq-sys}). If $x_3=0$, then $x_1=x_2$ and (\ref{eq-sys}) has $p^n$ solutions. Now suppose $x_3\neq 0$, let $y_1=\frac{x_1}{x_3}$ and $y_2=\frac{x_2}{x_3}$, then $y_1$ and $y_2$ satisfy
\begin{eqnarray}\label{L3-eq1}
\left\{
\begin{array}{lllll}
y_1-y_2+1=0,\\
y^{p^n-3}_1-y^{p^n-3}_2+1=0.
\end{array} \right.\ \
\end{eqnarray}
Denote by $L_0$ the number of solutions $(y_1,y_2)\in(\gf_{p^n})^2$ of (\ref{L3-eq1}). Thus we have $|\ffm_4|=p^n+(p^n-1)L_{0}$.

Note that (\ref{L3-eq1}) is equivalent to
\begin{equation}\label{L3-eq2}
(y_1+1)^{p^n-3}-y^{p^n-3}_1=1.
\end{equation}
It is obvious that $y_1=0$ is a solution of (\ref{L3-eq2}). If $y_1\neq 0$, then (\ref{L3-eq2}) is equivalent to $g_1(x)=0$, which has been  investigated in  Lemma \ref{Lem-g1}.
Thus, $L_0=1+T_1$ and $|\ffm_4|=p^n+(1+T_1)(p^n-1)$. This proves Lemma \ref{L3-Mi}.
\end{proof}

\begin{lemma} \label{Lem-M0} With the notation introduced above, we have
\begin{eqnarray*} \label{M0:eq3}
\left|\ffm^\circ\right|=(p^n-1)\left(3p^n-8-2\eta(-1)-\eta(-3)\left(2+\eta(-3)\right)+\lambda_{2,p^n}\right),
\end{eqnarray*} where $\lambda_{2,p^n}$ is defined as in (\ref{2:f2}).
\end{lemma}

\begin{proof}
Since $x_4\neq0$, putting $y_i=\frac{x_i}{x_4}$ for $i=1,2$ and $3$, we have
\begin{eqnarray}\label{Lem-M0-eq1}
\left\{
\begin{array}{lllll}
y_1-y_2+y_3-1=0,\\
y^{p^n-3}_1-y^{p^n-3}_2+y^{p^n-3}_3-1=0.
\end{array} \right.
\end{eqnarray}
Denote by $M_0$ the number of solutions $(y_1,y_2,y_3) \in \left(\gf_{p^n}^*\right)^3$ of the equation system (\ref{Lem-M0-eq1}). Then we have
\begin{eqnarray} \label{M-and-M0} \left|\ffm^\circ\right|=M_0(p^n-1).\end{eqnarray}
Now we compute $M_0$.

Since $y_i\ne 0$ for all $i \in \{1,2,3\}$, using $y_1y_3=z$, then (\ref{Lem-M0-eq1}) becomes
\begin{eqnarray}\label{Lem-M0-eq2}
\left\{
\begin{array}{lllll}
y_1+y_3=1+y_2,  \\
y_1y_3=z, \quad z \in \gf_{p^n}^*, \\
y^{-2}_1+y^{-2}_3=1+y^{-2}_2.
\end{array} \right.
\end{eqnarray}
From the second and the third equations in (\ref{Lem-M0-eq2}) we get
\[y^{-2}_2+1=\frac{y^2_1+y^2_3}{y^2_1y^2_3}=\frac{(y_2+1)^2-2z}{z^2},\]
which is equivalent to
\begin{eqnarray*}
(y^{-2}_2+1)z^2+2z-(y_2+1)^2=0.
\end{eqnarray*}
Then, we can conclude that $M_0$ is equal to the number of solutions $(y,y_2,z) \in \left(\gf_{p^n}^*\right)^3$ of the equation system
\begin{eqnarray}\label{Lem-M0-eq3}
\left\{
\begin{array}{lll}
y^2-(1+y_2)y+z=0,  \\
(y^{-2}_2+1) z^2+2z-(y_2+1)^2=0.
\end{array} \right.
\end{eqnarray}
 For determining $M_0$, now our strategy is to count the number of pairs $(y,z) \in \left(\gf_{p^n}^*\right)^2$ satisfying (\ref{Lem-M0-eq3}) for each fixed $y_2 \in \gf_{p^n}^*$.  We distinguish two cases as follows.

\textbf{\emph{Case 1:}} $y^{-2}_2+1=0$.  This case occurs only when $\eta(-1)=1$.  Then $y_2=\pm \sqrt{-1}$ and it follows that $z=y_2$ from the second equation in (\ref{Lem-M0-eq3}). Then the first equation in (\ref{Lem-M0-eq3}) leads to $y=1$ or $y=y_2$. Thus, for each such $y_2$ it contributes $2$ solutions to $M_0$.

\textbf{\emph{Case 2:}} $y^{-2}_2+1\neq0$. Then,  the second equation in (\ref{Lem-M0-eq3}) is a quadratic equation in variable $z$, and it has two solutions $z=y_2$ and $z=-\frac{y_2(y_2+1)^2}{y^2_2+1}$. There are two subcases
that need to  be considered.

\textbf{\emph{Subcase 2.1:}} $z=y_2$. Then the first equation in (\ref{Lem-M0-eq3}) still has two solutions $y=1$ or $y=y_2$ if $y_2 \ne 1$; however, it leads to one solution if $y_2=1$.

\textbf{\emph{Subcase 2.2:}} $z=-\frac{y_2(y_2+1)^2}{y^2_2+1}$. Then the first equation of (\ref{Lem-M0-eq3}) becomes
\begin{eqnarray}\label{Lem-M0-eq4}
y^2-(y_2+1)y-\frac{y_2(y_2+1)^2}{y^2_2+1}=0.
\end{eqnarray}
This is a quadratic equation in variable $y$ with discriminant given by  $\Delta=\frac{(y_2+1)^2(y^2_2+4y_2+1)}{y^2_2+1}$. Note that $y_2=-1$ will leads to a zero solution $y=0$ of (\ref{Lem-M0-eq4}) and $z=0$, which should be discarded since $y$, $z\in\gf_{p^n}^*$. Therefore, $y_2\neq -1$. When $y_2\in\gf_{p^n}\setminus\{0,-1\}$, (\ref{Lem-M0-eq4}) has two solutions in $\gf_{p^n}^*$ if $\eta(\Delta)=1$, a unique solution in $\gf_{p^n}^*$ if $\eta(\Delta)=0$, and no solution if $\eta(\Delta)=-1$. So this subcase contributes $\left(1+\eta\left(\frac{y^2_2+4y_2+1}{y^2_2+1}\right)\right)$ solutions for $y_2 \in \gf_{p^n} \setminus \{0,-1\}$.

Note that when $y_2=-\frac{y_2(y_2+1)^2}{y^2_2+1}$, the solutions in Subcases 2.1 and 2.2 will overlap, and they need to be excluded in the counting. Since $y_2\neq0$, $y_2=-\frac{y_2(y_2+1)^2}{y^2_2+1}$ is equivalent to that $y_2^2+y_2+1=0$. This holds if and only if when $\eta(-3)=1$ or $\eta(-3)=0$, and for such $y_2$ the above two subcases are the same. More precisely, when $\eta(-3)=1$, we can solve $y_2=\frac{-1\pm\sqrt{-3}}{2} \in \gf_{p^n}^*$ and each $y_2$ contributes $2$ solutions to $M_0$; if $\eta(-3)=0$, i.e., $p=3$, then $y_2=1$ and it contributes only one solution to $M_0$; if $\eta(-3)=-1$, then no such $y_2$ exists in $\gf_{p^n}$. Therefore, there are $\sum_{\substack{y_2^{2} +y_2+1 = 0}}\left(1+\eta(-3)\right)$ solutions that have been counted twice in Subcases 2.1 and 2.2.

Summarizing the above discussions, we can write the total number of solutions $M_0$ of the equation system (\ref{Lem-M0-eq3}) as
\[
\sum_{\substack{y_2^{-2} +1 = 0}}2+
\sum_{\substack{y_2^{-2} +1 \ne 0\\
y_2 \ne 1,0}}2+
\sum_{\substack{y_2^{-2} +1 \ne 0\\
y_2=1}}1+
\sum_{\substack{y_2^{-2} +1 \ne 0\\
y_2 \ne 0,-1}} \left(1+\eta\left(\frac{y_2^2+4y_2+1}{y_2^2+1}\right)\right)-\sum_{\substack{y_2^{2} +y_2+1 = 0}}\left(1+\eta(-3)\right). \]
Noting that
\[\sum_{\substack{y_2^{-2} +1 = 0}}2=2(1+\eta(-1)), \quad \sum_{\substack{y_2^{2} +y_2+1 = 0}}\left(1+\eta(-3)\right)=\left(1+\eta(-3)\right)^2,\]
and \[\sum_{y_2^{-2} +1=0} \eta\left(\frac{y_2^2+4y_2+1}{y_2^2+1}\right)=0,\]
we can obtain
\begin{eqnarray*}  M_0=3p^n-8-2\eta(-1)-\eta(-3)\left(2+\eta(-3)\right)+\sum_{y_2 \in \gf_{p^n}} \eta\left(\frac{y_2^2+4y_2+1}{y_2^2+1}\right). \end{eqnarray*}
 Then the desired value of $|\ffm^\circ|$ follows immediately from the fact that $\lambda_{2,p^n}=\sum_{y_2 \in \gf_{p^n}} \eta\left(\frac{y_2^2+4y_2+1}{y_2^2+1}\right)$ and the relation (\ref{M-and-M0}).
\end{proof}

With the above preparations, we can now obtain the value of $M$ easily.
\begin{theorem}\label{L4}   Let $p\geq 3$ be an odd prime, $T_1$ be given in Lemma \ref{Lem-g1}, and $\lambda_{2,p^n}$ be defined as in (\ref{2:f2}). Then the number of solutions to the equation system (\ref{eq-sys}), denoted by $M$,  is given by
$$M=1+(p^n-1)\left[3p^n+\lambda_{2,p^n}+4T_1-4-2\eta(-1)-\eta(-3)\left(2+\eta(-3)\right)\right].$$
\end{theorem}
\begin{proof} By the inclusion-exclusion principle we have
\[M=|\ffm^\circ|+\sum_{i=1}^4 \left|\ffm_i\right|-\sum_{1 \le i<j \le 4} \left|\ffm_i \cap \ffm_j\right|+\sum_{1 \le i<j <k\le 4} \left|\ffm_i \cap \ffm_j\cap \ffm_k\right|-\left|\cap_{i=1}^4\ffm_i\right|. \]
Then using Lemmas \ref{L3-Mi} and \ref{Lem-M0} and noting (\ref{M:eq1}) and (\ref{M:eq2}), we obtain the desired result.
\end{proof}

%\begin{remark}
%The above theorem shows that the value of $M$ can be derived from the quadratic character sum $\lambda_{2,p^n}$, which has been investigated in Section
%\ref{pre}.  That is to say, for any prime $p\geq3$,
%one can explicitly express the parameter $M$ in terms of $n$ by utilizing (\ref{pre:p=3}), Theorem \ref{thm gama}, Lemma \ref{Lem-g1} and Theorem \ref{L4}.
%\end{remark}

\section{The differential spectrum of $x^{p^n-3}$}\label{sec4}
For the power function $F(x)=x^{p^n-3}$ with $p$ being an odd prime in Theorem \ref{hellresult}, it is already known
that the differential uniformity $\delta(F)$ of $F(x)$ satisfies $1\leq\delta(F)\leq 5$ \cite{HRS}.  Recalling Definition \ref{def1}, we can assume the differential spectrum of $F(x)=x^{p^n-3}$ as
$$\mathbb{S}=[\omega_0,\omega_1,\omega_2,\omega_3,\omega_4,\omega_5 ]. $$
For $p=3$, the differential spectrum $\mathbb{S}$ has been completely  determined in \cite{XZLH}. However, the method used in \cite{XZLH} heavily depends on the characteristic $p=3$ and doesn't  seem to work for the general case  $p\geq5$.  In this section, for any odd prime $p\geq 3$, we will
compute  $\mathbb{S}$ by a unified approach.

\subsection{Some basic properties about the differential spectrum}
Before beginning  our computations, we mention some basic properties about the differential spectrum of a power mapping $x^d$ over finite fields.  Let  $x^d$ be a power mapping over $\gf_{p^n}$ with differential uniformity $\delta$, then using the notation in Definition \ref{def1} we have \begin{equation}\label{identity}
\sum_{i=0}^{\delta}\omega_i=\sum_{i=0}^{\delta}i\omega_i=p^n.
\end{equation}
The identities in  (\ref{identity}) are well-known \cite{BCC,Yan,XZLH}, and are useful in computing the differential spectrum. Moreover, the following identity
also plays an important role in the computation, which was established in \cite{HRS}.

\begin{lemma}\cite[Theorem 10]{HRS}
With the notation introduced in Definition \ref{def1}, let $M$ denote the number of solutions $(x_1,x_2,x_3,x_4)\in (\gf_{p^n})^4$ of the equation system
\begin{eqnarray*}
\left\{
\begin{array}{lllll}
x_1-x_2+x_3-x_4&=&0,\\
x^{d}_1-x^{d}_2+x^{d}_3-x^{d}_4&=&0.
\end{array} \right.\ \
\end{eqnarray*}
Then, we have \begin{equation}\label{i2omega}
\sum_{i=0}^{\delta}i^2\omega_i=\frac{M-p^{2n}}{p^n-1}.
\end{equation}
\end{lemma}

With the equalities in \eqref{identity} and \eqref{i2omega}, our strategy for computing the differential spectrum $\mathbb{S}$ of $x^{p^n-3}$ can be sketched as follows: first we will compute $\omega_5$, $\omega_3$  and $\omega_2$; then we  establish a system of linear equations in three variables $\omega_0$, $\omega_1$ and $\omega_4$ by (\ref{identity}) and (\ref{i2omega}), which enables us to express $\omega_0$, $\omega_1$ and $\omega_4$ in terms of the known $\omega_5$, $\omega_3$  and $\omega_2$. Next we begin with the general setup for investigating the differential spectrum.

\subsection{The general setup}
 For any $b\in\gf_{p^n}$, the derivative equation $\mathbb{D}_1F(x)=b$ is
\begin{eqnarray}\label{main}
(x+1)^{p^n-3}-x^{p^n-3}=b.
\end{eqnarray}
Let $N(b)$ denote the number of its solutions in $\gf_{p^n}$. The elements $\omega_i$'s for $i\in\{0,1,\cdots,5\}$ in
the differential spectrum $\mathbb{S}$  are actually the number of  $b\in\gf_{p^n}$ such that $N(b)=i$.

It can be easily observed  that the derivative equation \eqref{main} has a solution $x$ if and only if the derivative equation $(x+1)^{p^n-3}-x^{p^n-3}=-b$ has a solution $-x-1$. Thus,  $N(b)=N(-b)$ for any $b$.  When $b=0$, it is easy to verify that $x=-\frac{1}{2}$ is the unique solution of (\ref{main}). That is ta say,  $N(0)=1$. Moreover,  note that in  (\ref{main}) if  $b$ is equal to $1$ (resp. $-1$), then $x=0$ (resp. $x=-1$) is a solution to the corresponding equation (\ref{main}). Since $N(0)$ is already determined, in the following we only need to consider $N(b)$ for $b\neq0$.

  For $b\in\gf_{p^n}^*$, define
\begin{eqnarray}\label{poly-gb}
g_{b}(x)=x^4+2x^3+x^2+2b^{-1}x+b^{-1},
 \end{eqnarray}
and denote the number of its roots in $\gf_{p^n}$ by $T_b$. Note that for $b=1$, $T_b$ has already been determined in Lemma \ref{Lem-g1}.
This polynomial is closely connected with the derivative equation \eqref{main}.
As a matter of fact, when $x\neq0, -1$, (\ref{main})  can be written as $(x+1)^{-2}-x^{-2}=b$, which is equivalent to
\begin{eqnarray*}
	g_b(x)=x^4+2x^3+x^2+2b^{-1}x+b^{-1}=0. \end{eqnarray*}
Hence we can arrive at the following result:
\begin{equation}\label{T-and-N}
N(b)=T_b\,\,\,\,\mbox{for\,\,\,\,each\,\,\,\,} b\in\gf_{p^n}^*\setminus\{\pm 1\} \,\,\,\,\mbox{and} \,\,\,\,N(\pm1)=1+T_{\pm1}.
\end{equation}
Moreover, since $N(b)=N(-b)$ for any $b$, it follows that \begin{equation}\label{pofTb}
T_b=T_{-b} \text{ for any }b\in\gf_{p^n}^*.\end{equation}

%\smallskip

\subsection{The values of $\omega_5$}

Note that (\ref{poly-gb}) has at most four roots in $\gf_{p^n}$. By  (\ref{T-and-N}), it is easy to see that $\delta(F)=5$ if and only if $N(1)=N(-1)=5$. Then, we have $\omega_5 \in \{0,2\}$, and $\omega_5=2$ if and only if $T_1=4$. The condition for $T_1=4$ has already been shown in Lemma \ref{Lem-g1}. Thus, we can determine $\omega_5$ in the differential spectrum $\mathbb{S}$ as follows.
\begin{theorem}\label{w5}With the notation introduced above,  we have
\begin{eqnarray*}
\omega_5=\left\{
\begin{array}{lllll}
2, ~\mathrm{if}~\eta(2)=\eta(-7)=\eta(-1+2\sqrt{2})=1,\\
0, ~\mathrm{otherwise}.\\
\end{array} \right.\ \
\end{eqnarray*}
\end{theorem}

%\begin{remark}
%For each given prime $p\ge 3$, when applying this theorem to obtain $\omega_5$, one first needs to check whether the element $2$ is a square in $\gf_{p^n}$ or not. If $\eta(2)=1$, then one also needs to compute $\eta(-1\pm 2\sqrt{2})$
%{\mli as discussed in Remark \ref{remark of T1}}
%\end{remark}

\subsection{The values of $\omega_3$}

Next we investigate the value of $\omega_3$. When $N(b)=3$, we distinguish the following two cases.

{\it Case 1: } $b=\pm1$.  $N(1)=N(-1)=3$ if and only if  $T_1=2$. By Lemma \ref{Lem-g1}, this occurs only when $\eta(2)=1$ and $\eta(-7)=-1$.

{\it Case 2: } $b\neq\pm1$. By (\ref{T-and-N}),  $N(b)=3$ if and only if $T_b=3$. Thus, we need to characterize when $T_b=3$ for $b\in\gf_{p^n} \setminus \{0,\pm 1\}$. Such results are given below.

\begin{lemma}\label{repeat} Let $b\in\gf_{p^n} \setminus \{0,\pm 1\}$, and $g_b(x)$ be the polynomial defined as in (\ref{poly-gb}).  Then $g_b(x)=0$ has a multiple root $x_0\in\gf_{p^n}$ if and only if $p \ne 7$ and $\eta(-3)=1$. In this case, the multiple roots $x_0$'s are $-\frac{1}{2}\pm \frac{1}{6}\sqrt{-3}$,  and the corresponding $b$'s are $\mp 3\sqrt{-3}$. \end{lemma}
\begin{proof}
If $x_0$ is a multiple root of  $g_b(x)=0$, then $g^{\prime}_b(x_0)=2(2x^3_0+3x^2_0+x_0+b^{-1})=0$, and we have $x_0\neq0,-1$. Hence $b^{-1}=-(2x^3_0+3x^2_0+x_0)$.   Substituting it into the original equation,  we get $$x_0(x_0+1)(3x_0^2+3x_0+1)=0.$$
This together  with $x_0\neq0,-1$ leads to $3x^2_0+3x_0+1=0$. Such $x_0$ exists if and only if $\eta(-3)=1$. Then we have $x_0=-\frac{1}{2}\pm\frac{1}{6}\sqrt{-3}$ and the corresponding $b$'s are $\mp3\sqrt{-3}$.

Moreover, if $p=7$,  then we may take $\sqrt{-3}=2$ since $2^2=-3$, and thus $b=\mp 3 \sqrt{-3}=\pm 1$, a contradiction. Therefore, we need the condition $p\neq 7$ holds.
\end{proof}

\begin{lemma} \label{3:cor1} Let $b\in\gf_{p^n} \setminus \{0,\pm 1\}$. Then $T_b=3$ if and only if $p \ne 7$, $\eta(-3)=\eta(-2)=1$ and $b=\pm3\sqrt{-3}$.
\end{lemma}
\begin{proof}If $T_b=3$, then $g_b(x)=0$ must have a multiple root $x_0$ in $\gf_{p^n}$. By Lemma \ref{repeat}, we have $p \ne 7$, $\eta(-3)=1$ and $(x_0,b)=(-\frac{1}{2}+\frac{1}{6}\sqrt{-3},-3\sqrt{-3})$ or $(x_0,b)=(-\frac{1}{2}-\frac{1}{6}\sqrt{-3},3\sqrt{-3})$. For the former case, the equation $g_b(x)=0$ can be written as
\begin{eqnarray} \label{3:eq1} (x-x_0)^2\left(x^2+(1+\frac{1}{3}\sqrt{-{3}})x+(-\frac{1}{2}+\frac{1}{6}\sqrt{-{3}})\right)=0,\end{eqnarray}
which has three solutions in $\gf_{p^n}$ if and only if $\eta((1+\frac{1}{3}\sqrt{-{3}})^2-4(-\frac{1}{2}+\frac{1}{6}\sqrt{-{3}}))=\eta(\frac{8}{3})=1$, that is, $\eta(6)=1$. Since $\eta(-3)=1$, this is equivalent to that $\eta(-2)=1$. It can be checked that in this case the other two solutions of (\ref{3:eq1}) are $-\frac{1}{2}-\frac{1}{6}\sqrt{-3}\pm\frac{1}{3}\sqrt{6}$, which are different from $x_0$. For the latter case $(x_0,b)=(-\frac{1}{2}-\frac{1}{6}\sqrt{-3},3\sqrt{-3})$, the arguments are almost the same. So we omit the details.
\end{proof}

Based on the above results, we can now obtain the value of $\omega_3$ below.
\begin{theorem}\label{w3}Let ${\rm C}_1$ denote the condition that $\eta(2)=-\eta(-7)=1$ and  ${\rm C}_2$ denote  the condition  that $\eta(-3)=\eta(6)=1$ and $p\neq 7$. Then, we have
\begin{eqnarray} \label{3:o31}
\omega_3=\left\{
\begin{array}{ll}
4, & ~\mathrm{both~} {\rm C}_1 \mathrm{ ~and}~{\rm C}_2 \mathrm{ ~hold},\\
2, & ~\mathrm{only ~one ~of ~} {\rm C}_1 \mathrm{ ~and}~{\rm C}_2 \mathrm{ ~holds},\\
0, & ~\mathrm{otherwise}.
\end{array} \right.
\end{eqnarray}
Alternatively, the value $\omega_3$ may be expressed as
\begin{eqnarray} \label{3:o3} \omega_3=\frac{\eta(7)^2\eta(3)^2 }{2} \Big(\big(1+\eta(2)\big) \cdot \big(1-\eta(-7)\big)+\big(1+\eta(-2) \big) \cdot \big(1+\eta(-3)\big)\Big).\end{eqnarray}
\end{theorem}
\begin{proof} In order to find the value of $\omega_3$, we need to find the frequency of $b \in \gf_{p^n}$ such that $N(b)=3$. There are two cases to consider.

{\it Case 1. $b= \pm1$:} $N(1)=N(-1)=3$ if and only if  $T_1=2$. By Lemma \ref{Lem-g1}, this occurs if and only if $\eta(2)=1$ and $\eta(-7)=-1$, which is Condition ${\rm C}_1$.

{\it Case 2. $b \neq \pm1$:} By Lemma \ref{3:cor1}, $N(b)=T_b=3$ if and only if $p \ne 7,\,\,\,\, \eta(-3)=\eta(-2)=1$ and the corresponding $b'$s are $\pm3\sqrt{-3}$. Here we get Condition ${\rm C}_2$.

Combining these two cases yield the expression of $\omega_3$ in (\ref{3:o31}). As for the expression of $\omega_3$ in (\ref{3:o3}), denote
\begin{eqnarray*} f_1:&=&\eta(7)^2 \eta(3)^2\cdot \big(1+\eta(2)\big) \cdot \big(1-\eta(-7)\big), \\ f_2:&=&\eta(7)^2 \eta(3)^2 \cdot \big(1+\eta(-2) \big) \cdot \big(1+\eta(-3)\big). \end{eqnarray*}
Then (\ref{3:o3}) follows easily from the observation that
\[f_1=\left\{\begin{array}{ll}
4,& ~\mathrm{if~} {\rm C}_1 \mathrm{ ~holds},\\
0,& ~\mathrm{if~} {\rm C}_1 \mathrm{ ~does ~ not ~ hold},
\end{array}\right. \quad f_2=\left\{\begin{array}{ll}
4,& ~\mathrm{if~} {\rm C}_2 \mathrm{ ~holds},\\
0,& ~\mathrm{if~} {\rm C}_2 \mathrm{ ~does ~ not ~ hold},
\end{array}\right.\]
This finishes the proof of Theorem \ref{w3}.
\end{proof}

\subsection{The value of $\omega_2$}
This subsection is devoted to the computation of $\omega_2$.  Recall the basic facts in (\ref{T-and-N}) and (\ref{pofTb}). First, we prove the following useful result.

\begin{lemma}\label{A-and-B} Let $p\geq 3$ and let $T_b$ be the number of roots of the polynomial $g_b(x)\in \gf_{p^n}[x]$ defined as in (\ref{poly-gb}). Define
two sets
\begin{equation}\label{defA}
\mathcal{A}=\{a\in\gf_{p^n}\,\,\mid\,\,\eta(a^2-4)=1\,\,\,\,\mbox{and}\,\,\,\,\eta(-3a^2-4)=-1\},
\end{equation}
and
\begin{equation}\label{defB}
\mathcal{B}=\{b\in \gf_{p^n}^*\,\,\mid\,\,T_b=2\}.
\end{equation}Then, there is a one-to-one correspondence between the elements $b\in \mathcal{B}$ and the elements $a\in \mathcal{A}$. Moreover,
if $\eta(2)=1$ and $\eta(-7)=-1$, then $\pm1 \in \mathcal{B}$ and the corresponding $a$'s belong to $\{\pm 2\sqrt{2}\}$.\end{lemma}
\begin{proof}
For $b \in \mathcal{B}$, the proofs of Lemmas \ref{repeat} and \ref{3:cor1} imply that $g_b(x)=0$ can not have multiple roots, so it has exactly two distinct simple roots in $\gf_{p^n}$. Putting $y=2x+1$ in (\ref{poly-gb}), $g_b(x)=0$ becomes
\begin{eqnarray}\label{eqa-t-y}
y^4-2y^2+16b^{-1}y+1=0,
\end{eqnarray}
which also has exactly two distinct simple roots in $\gf_{p^n}$ for each $b\in \mathcal{B}$. Thus, we can factor (\ref{eqa-t-y}) into the form
\begin{equation}\label{y1y2f}
(y^2+ay+c)(y^2-ay+c^{-1})=0,
\end{equation}
where the pair $(a,c)$ satisfies the following conditions
\begin{enumerate}
\item $a\in \gf_{p^n}^*$, $c\in\gf_{p^n}^*$;

\item $y^2+ay+c$ is irreducible over $\gf_{p^n}$, that is, $\eta(a^2-4c)=-1$;

\item $y^2-ay+c^{-1}$ has two distinct roots in $\gf_{p^n}$, that is, $\eta(a^2-4c^{-1})=1$;

\item
%comparing  (\ref{eqa-t-y}) with (\ref{y1y2f}), we also have
\begin{eqnarray}\label{acb}
\left\{
\begin{array}{lllll}
c+c^{-1}&=a^2-2,\\
a(c-c^{-1})&=-16b^{-1},\\
\end{array} \right.\ \
\end{eqnarray}
\end{enumerate}
which is obtained by comparing  (\ref{eqa-t-y}) with (\ref{y1y2f}).
This gives the correspondence from $b\in \mathcal{B}$ to the pairs $(a,c)$ satisfying the above conditions. Once $b\in \mathcal{B}$ is given, the two solutions of
(\ref{eqa-t-y}) are uniquely determined and so are the pair $(a,c)$ and the element $a$. This shows that for each $b\in \mathcal{B}$, there  exists a unique $a$ satisfying the conditions in $1)-4)$. Now we verify that this $a\in \mathcal{A}$. For this $a$, (\ref{acb}) implies that $c+c^{-1}=a^2-2$ has two distinct roots $c \ne c^{-1} \in \gf_{p^n}$, so we have $\eta((a^2-2)^2-4)=\eta(a^2(a^2-4))=1$, that is, $\eta(a^2-4)=1$ and $a\in \gf_{p^n}^*$. On the other hand, from
\begin{equation}\label{disc-mult}-1=\eta((a^2-4c)(a^2-4c^{-1}))=\eta(a^4-4a^2(c+c^{-1})+16)=\eta((a^2-4)(-3a^2-4)),\end{equation}
we derive that $\eta(-3a^2-4)=-1$. This shows that $a$ indeed belongs to $\mathcal{A}$.

Now suppose that $a \in \mathcal{A}$. We show that $c$ and $b$ are all uniquely determined by this $a$ according to (\ref{y1y2f}), and $b \in \mathcal{B}$.

First, since $\eta(a^2-4)=1$, the first equation of (\ref{acb}) has two distinct solutions $c_1, c_2 \in \gf_{p^n}$, and we have $\eta((a^2-4c_1)(a^2-4c_2))=-1$ due to (\ref{disc-mult}). We may assume $\eta(a^2-4c_1)=-1$. Then we take $c=c_1$. This is the desired $c$ in (\ref{y1y2f}) such that $y^2+ay+c$ is irreducible over $\gf_{p^n}$ and $y^2-ay+c^{-1}$ is reducible with two distinct roots in $\gf_{p^n}$. Choosing $b$ according to the second equation of (\ref{acb}), we can obtain
\[ y^4-2y^2+16b^{-1}y+1=(y^2+ay+c)(y^2-ay+c^{-1})=0,\]
which has exactly two roots in $\gf_{p^n}$. This shows that $b \in \gf_{p^n}^*$ is uniquely determined by $a$ and it satisfies $T_b=2$. This finishes the proof of the first part of Lemma \ref{A-and-B}.

As for the second part, when $\eta(2)=1$ and $\eta(-7)=-1$, by Lemma \ref{Lem-g1} and (\ref{pofTb}), we have $T_1=T_{-1}=2$ and thus $\pm 1\in \mathcal{B}$.   Then from  (\ref{acb}) we obtain
\[16^2=(-16b^{-1})^2=a^2(c-c^{-1})^2=a^2((a^2-2)^2-4),\]
that is,
 \begin{equation*}\label{ex-equ} (a^2-8)(a^4+4a^2+32)=0,\end{equation*}
which implies that $a^4+4a^2+32=0$ or $a^2=8$. If $a^4+4a^2+32=0$, then we have $\eta(4^2-4 \cdot 32)=\eta(-7)=1$, a contradiction. Hence we have $a^2=8$ and it can be easily verified that the corresponding two $a$'s indeed belong to $\mathcal{A}$. This proves the second part of Lemma \ref{A-and-B}.
\end{proof}

Now we can obtain the value of $\omega_2$ in the following theorem.

\begin{theorem} \label{w2} For $p \ge 3$, we have
\begin{eqnarray*}
\omega_2=\left\{
\begin{array}{cl}
0,&~ \mathrm{if}~p=3~\mathrm{and~}n~\mathrm{is~even},\\
\frac{3^n-3}{2},&~ \mathrm{if}~p=3~\mathrm{and~}n~\mathrm{is~odd},\\
A+2,&~ \mathrm{if}~p=7~\mathrm{and~}n~\mathrm{is~odd},\\
A-2,&~ \mathrm{if}~\eta(2)=1~\mathrm{and~}\eta(-7)=-1,\\
A, &~\mathrm{otherwise}.\\
\end{array} \right.\ \
\end{eqnarray*}
where $
A=\frac{1}{4}\big(p^n-5-\lambda_{1,p^n}-\eta(-3)+2\eta(-1)\big)
$ with $\lambda_{1,p^n}$ being defined as in (\ref{2:f1}).
\end{theorem}

\begin{proof}
If $N(1)=N(-1)=2$, then $T_1=T_{-1}=1$. By Lemma \ref{Lem-g1}, this occurs only when $p=7$ and $n$ is odd.
Now we need to consider the number of $b\in\gf_{p^n}\setminus \{0,\pm1\}$ such that $T_b=2$, by Lemma \ref{A-and-B} which is related to the cardinality of the set $\mathcal{A}$  defined in (\ref{defA}) or $\mathcal{B}$ in (\ref{defB}). We distinguish the following two cases:

\textbf{\emph{Case 1}:} $p=3$. Then, we have $T_{\pm1}\neq 1$, $\pm 1\notin \mathcal{B}$ and the set $\mathcal{A}$ defined in (\ref{defA}) becomes  \begin{eqnarray*}
\mathcal{A}=\{a\in\gf_{3^n}~|~\eta(a^2-1)=1~\mathrm{and}~\eta(-1)=-1\}.\end{eqnarray*}
Therefore, in this case we have $$\omega_2=|\mathcal{B}|=|\mathcal{A}|.$$
If $n$ is even, then $\eta(-1)=1$ and thus $|\mathcal{A}|=0$. Otherwise, we have
\begin{eqnarray*}
\mathcal{A}=\{a\in\gf_{3^n}^*~|~\eta(a^2-1)=1\},\end{eqnarray*}
and by the cyclotomic numbers used in \cite{XZLH}, we have $$|\mathcal{A}|=\frac{3^n-3}{2}.$$
Thus, in this case, we have
\begin{eqnarray*}
\omega_2=\left\{
\begin{array}{cl}
0,&~ \mathrm{if}~n~\mathrm{is~even},\\
\frac{3^n-3}{2},&~\mathrm{if}~~n~\mathrm{is~odd}.\\
\end{array} \right.\ \
\end{eqnarray*}

\textbf{\emph{Case 2}:}  $p\geq 5$. Then, by Lemma \ref{A-and-B} we find that
\[\left|\{b \in \gf_{p^n} \setminus \{0, \pm 1\}: T_b=2\}\right| =\left\{
\begin{array}{lllll}
|\mathcal{A}|-2, ~ \mathrm{if}~\eta(2)=1~\mathrm{and~}\eta(-7)=-1,\\
|\mathcal{A}|, ~\mathrm{otherwise}.
\end{array}\right.\]
This shows that when $p\geq 5$

\begin{eqnarray}\label{w2-derive}
\omega_2=\left\{
\begin{array}{ll}
|\mathcal{A}|-2, ~ \mathrm{if}~\eta(2)=1~\mathrm{and~}\eta(-7)=-1,\\
|\mathcal{A}|+2, ~ \mathrm{if}~p=7~\mathrm{and~}n~\mathrm{is~odd},\\
|\mathcal{A}|, ~\mathrm{otherwise}.
\end{array} \right.\ \
\end{eqnarray}
Now it suffices to determine the cardinality of $\mathcal{A}$. Actually, since $p \ge 5$,
\begin{equation*}\begin{array}{lcl}|\mathcal{A}|
%&=\frac{1}{4}\sum\limits_{a\in\gf_{p^n}}\left(1+\eta(a^2-4)\right)\left(1-\eta(-3a^2-4)\right)-\frac{1}{2}\left(1-\eta(-1)\right)-\frac{1}{2}\left(1+\eta(-3)\right)\\
&=&\frac{1}{4}\sum\limits_{a\in\gf_{p^n},a^2\neq4,-\frac{4}{3}}\left(1+\eta(a^2-4)\right)\left(1-\eta(-3a^2-4)\right)\\
&=&\frac{1}{4}\sum\limits_{a\in\gf_{p^n}}\left(1+\eta(a^2-4)\right)\left(1-\eta(-3a^2-4)\right)-1+\frac{1}{2}\eta(-1)-\frac{1}{2}\eta(-3)\\
&=&\frac{1}{4}\big[\sum\limits_{a\in\gf_{p^n}}1+\sum\limits_{a\in\gf_{p^n}}\eta(a^2-4)-\sum\limits_{a\in\gf_{p^n}}\eta(-3a^2-4)\\
&&-\sum\limits_{a\in\gf_{p^n}}\eta((a^2-4)(-3a^2-4))\big]-1+\frac{1}{2}\eta(-1)-\frac{1}{2}\eta(-3).
\end{array}
\end{equation*}
By using $\sum\limits_{a\in\gf_{p^n}}\eta(a^2-4)=-1$, $\sum\limits_{a\in\gf_{p^n}}\eta(-3a^2-4)=-\eta(-3)$ from Lemma \ref{charactersumquadratic}, and noting that $\lambda_{1,p^n}=\sum_{a\in\gf_{p^n}}\eta\left((a^2-4)(-3a^2-4)\right)$ which has been evaluated in Theorem \ref{thm gama}, we obtain $|\mathcal{A}|=A$.  Then the desired result follows from (\ref{w2-derive}).   This completes the proof of Theorem \ref{w2}.
\end{proof}

%\begin{remark}
%Theorem \ref{w2} gives the value of $\omega_2$ in the differential spectrum $\mathbb{S}$ of $x^{p^n-3}$. For $p=3$, the  value of $\omega_2$ is given explicitly in terms of $n$. For any $p\ge 5$,  the  value of $\omega_2$ is expressed in terms of the quadratic character sum $\lambda_{1,p^n}$ (see $A$ in Theorem \ref{w2}), which has been evaluated by the theory of elliptic curves in Theorem \ref{thm gama} and can be expressed in terms of $n$.
%\end{remark}

Based on the previous results and the identities in (\ref{identity}) and (\ref{i2omega}), we can obtain the following main result about the differential spectrum of $x^{p^n-3}$.

\begin{theorem}\label{spectrum}Let $\mathbb{S}=[\omega_0, \omega_1, \ldots, \omega_5]$ be the differential spectrum of $F(x)=x^{p^n-3}$. Then we have
\begin{eqnarray}\label{dfequation}
\left\{
\begin{array}{lllll}
\omega_0&=\frac{M-2p^{2n}+p^n}{4(p^n-1)}+\frac{1}{2}\omega_2+\frac{1}{2}\omega_3-\omega_5,\\
\omega_1&=\frac{-M+5p^{2n}-4p^n}{3(p^n-1)}-\frac{4}{3}\omega_2-\omega_3+\frac{5}{3}\omega_5,\\
\omega_4&=\frac{M-2p^{2n}+p^n}{12(p^n-1)}-\frac{1}{6}\omega_2-\frac{1}{2}\omega_3-\frac{5}{3}\omega_5,\\
\end{array} \right.\ \
\end{eqnarray}
where $\omega_5$, $\omega_3$ and $\omega_2$ are given in Theorems \ref{w5}, \ref{w3} and \ref{w2}, respectively, and $M$ is given in Theorem \ref{L4}.
\end{theorem}

\begin{remark}
	Applying Theorem \ref{spectrum}, the differential spectrum $\mathbb{S}$ of $x^{p^n-3}$  for any odd prime $p\geq 3$  can be completely determined. To be more concrete, for each given prime $p\geq3$, one first compute the exact values of $\omega_5$, $\omega_3$, $\omega_2$ and $M$:
\begin{itemize}
\item the values of $\omega_5$ and $\omega_3$ can be derived from Theorems \ref{w5} and \ref{w2} respectively after calculating the quadratic character of some specified elements;

\item the value of $\omega_2$ is given in Theorem \ref{w2}. For $p=3$, it is already given explicitly. For any $p\ge 5$,  $\omega_2$ is expressed in terms of the quadratic character sum $\lambda_{1,p^n}$, which has been evaluated in Theorem \ref{thm gama}.

\item the value of $M$ shown in Theorem \ref{L4} is related to the quadratic character sum $\lambda_{2,p^n}$.  For any prime $p\geq3$,
one can explicitly express the parameter $M$ in terms of $n$ by utilizing (\ref{pre:p=3}), Theorem \ref{thm gama}, Lemma \ref{Lem-g1} and Theorem \ref{L4}.
\end{itemize}
Then, the differential spectrum $\mathbb{S}$ can be computed  via (\ref{dfequation}), and one can express it explicitly in terms of $n$.
%It is worth noting that $\omega_2$ in Theorem \ref{w2} is dependent on the quadratic character sum $\lambda_{1,p^n}$ in the cases $p\ge 5$, and $M$ in Proposition \ref{L4} relies on $\lambda_{2,p^n}$ and $T_1$ from Lemma \ref{Lem-g1}.
\end{remark}

We provide the following results to illustrate Theorem \ref{spectrum}. The first one is about the case $p=3$, which has been investigated in \cite{XZLH} with a different method.

\begin{corollary}\label{p=3ds} Let $p=3$ and let $\mathbb{S} = [\omega_0, \omega_1, \ldots, \omega_5]$  be the differential spectrum of the power mapping $x^{p^n-3}$ over $\gf_{3^n}$. Then,

\noindent ({\rm i}) when $n$ is odd, $$
\mathbb{S}=
%\left[\frac{3^n-3}{2}, 3, \frac{3^n-3}{2}, 0, 0, 0\right];
\left[\omega_0=\frac{3^n-3}{2},\,\,\omega_1=3,\,\,\omega_2=\frac{3^n-3}{2},\,\,\omega_3=0,\,\omega_4=0,\,\,\omega_5=0\right];
$$

\noindent ({\rm ii}) when $n \equiv 2\,\,({\rm mod}\,\,4)$, $$\mathbb{S}=\left[\omega_0=\frac{3^n-9}{4},\,\,\omega_1=2\cdot3^{n-1}+3,\,\,\omega_2=0,\,\,\omega_3=0,\,\,\omega_4=\frac{3^{n-1}-3}{4},\,\,\omega_5=0\right];$$

\noindent ({\rm iii}) when $n \equiv 0\,\,({\rm mod}\,\,4)$, $$\mathbb{S}=\left[\omega_0=\frac{3^n-1}{4},\,\,\omega_1=2\cdot 3^{n-1}+1,\,\omega_2=0,\,\,\omega_3=0,\,\,\omega_4=\frac{3^{n-1}-11}{4},\,\,\omega_5=2\right].$$

\end{corollary}

\begin{proof} For $p=3$, by Theorem \ref{w5}, we have $\omega_5=0$ if $n$ is odd or $n \equiv 2\,\,({\rm mod}\,\,4)$, and $\omega_5=2$ if $n \equiv 0\,\,({\rm mod}\,\,4)$ since in this case $\left(x^2+1\right)^2-8=x^4+2x^2+2$ is irreducible over $\gf_{3}$. By (\ref{3:o3}) in Theorem \ref{w3}, we have $\omega_3=0$. By Theorem \ref{w2}, we have $\omega_2=0$ if $n$ is even, and $\omega_2=\frac{3^n-3}{2}$ if $n$ is odd. By Theorem \ref{L4}, we get $M=1+(3^n-1)(3^{n+1}-2)$ if $n$ is odd, $M=1+(3^n-1)(3^{n+1}-8)$ if $n \equiv 2\,\,({\rm mod}\,\,4)$, and $M=1+(3^n-1)(3^{n+1}+8)$ if $n \equiv 0\,\,({\rm mod}\,\,4)$. Then, we should distinguish three cases and substituting the corresponding values into (\ref{dfequation}), the differential spectrum $\mathbb{S}$ is derived.
\end{proof}

\begin{remark}

For the case $p=3$, based on the characteristic property, the work of \cite{XZLH}  calculated $\omega_4$ directly instead of investigating
 the parameter $M$. However, their method in \cite{XZLH}  doesn't seem to work for general case $p\ge 5$.
The approach in the present paper works for all odd primes.
\end{remark}

\begin{corollary}\label{p=5ds} Let $p=5$ and $\Gamma_{5,n}=-\left(-1 + 2 \sqrt{-1}\right)^n-\left(-1 - 2 \sqrt{-1}\right)^n$ obtained from Example \ref{examp1}. Then, the differential spectrum $\mathbb{S}$ of $x^{p^n-3}$ is shown  as follows:

\noindent ({\rm i}) when $n$ is odd, $\mathbb{S}$ is given by
$$\left[\omega_0=\frac{3\cdot 5^n+\Gamma_{5,\,n}-17}{8},\,\omega_1=\frac{5^n+10}{3},\,\omega_2=\frac{5^n-\Gamma_{5,\,n}-3}{4},\,\omega_3=0,\,\omega_4=\frac{5^n+3\cdot\Gamma_{5,\,n}-11}{24},\,\,\omega_5=0\right];$$

\noindent ({\rm ii}) when $n \equiv 2\,\,({\rm mod}\,\,4)$, $\mathbb{S}$ is given by $$\left[\omega_0=\frac{3\cdot 5^n+\Gamma_{5,\,n}-17}{8},\,\omega_1=\frac{5^n+8}{3},\,\omega_2=\frac{5^n-\Gamma_{5,\,n}-3}{4},\,\omega_3=2,\,\omega_4=\frac{5^n+3\cdot\Gamma_{5,\,n}-43}{24},\,\,\omega_5=0\right];$$

\noindent ({\rm iii}) when $n \equiv 0\,\,({\rm mod}\,\,4)$, $\mathbb{S}$ is given by $$\left[\omega_0=\frac{3\cdot 5^n+\Gamma_{5,\,n}-1}{8},\,\omega_1=\frac{5^n+2}{3},\,\omega_2=\frac{5^n-\Gamma_{5,\,n}-3}{4},\,\omega_3=2,\,\omega_4=\frac{5^n+3\cdot\Gamma_{5,\,n}-91}{24},\,\omega_5=2\right].$$

\end{corollary}

\begin{proof}If $p=5$, then $a=\Gamma_{5,\,1}=2$ and the explicit formula for $\Gamma_{5,\,n}$ follows from Theorem \ref{thm gama}. Next we consider the following two cases:

  {\it Case 1: } $n$ is odd. Then, the element $2$ is a nonsquare in $\gf_{5^n}$  since it is a nonsquare in $\gf_{5}$. Thus, we have $\omega_5=0$,  $\omega_3=0$ and  $\omega_2=\frac{5^n-\Gamma_{5,\,n}-3}{4}$ according to Theorems \ref{w5}, \ref{w3} and \ref{w2}, respectively. Furthermore, we have $T_1=0$ by Lemma \ref{Lem-g1} and $M=5^n+(5^n-1)(3\cdot 5^n+\Gamma_{5,\,n}-7)$ by Theorems \ref{L4} and \ref{thm gama}. Substituting $\omega_5$, $\omega_3$, $\omega_2$ and $M$ into Theorem \ref{spectrum}, we obtain the desired result.

  {\it Case 2: } $n$ is even. Then, the elements $\pm 2$ are squares in $\gf_{5^n}$. One needs to decide whether $-1\pm2\sqrt{2}$ are squares in $\gf_{5^n}$ or not. Note that $-1+2\sqrt{2}$ (resp. $-1-2\sqrt{2}$ ) is a square in $\gf_{5^n}$
  if and only if $(x^2+1)^2=8$ has a solution in $\gf_{5^n}$, while the associated polynomial $(x^2+1)^2-8$ is an irreducible polynomial over $\gf_{5}$. Thus, $-1+2\sqrt{2}$ (resp. $-1-2\sqrt{2}$ ) is a square in $\gf_{5^n}$ if and only if $n \equiv 0\,\,({\rm mod}\,\,4)$. As we have done in Case 1, the desired results then follows from Theorem \ref{spectrum}.
\end{proof}

Similarly, for $p=7$, the differential spectrum of the function $x^{7^n-3}$ over $\gf_{7^n}$ can be presented as follows.

\begin{corollary}\label{corop7}
The power mapping $x^{7^n-3}$ over $\gf_{7^n}$ is differentially $4$-uniform and  its  differential spectrum $\mathbb{S}$ is given as follows:

\noindent ({\rm i}) $\mathbb{S}=\left[\omega_0=\frac{3\cdot7^n-5}{8},\,\,\omega_{1}=\frac{7^n+2}{3},\,\,\omega_2=\frac{7^n+1}{4},\,\,\omega_{3}=0,\,\,\omega_{4}=\frac{7^n-7}{24}\right] $ if $n$ is odd;

\noindent ({\rm ii}) $\mathbb{S}=\left[\omega_0=\frac{3\cdot7^n-2(-7)^{n/2}-1}{8},\,\,\omega_{1}=\frac{7^n+2}{3},\,\,\omega_2=\frac{7^n+2(-7)^{n/2}-3}{4},\,\,\omega_{3}=0,\,\,\omega_{4}=\frac{7^n-6(-7)^{n/2}+5}{24}\right] $ if $n$ is even.

\end{corollary}

For other given primes $p$, one can also obtain similar results as Corollaries \ref{p=3ds}, \ref{p=5ds} and \ref{corop7} by Theorem \ref{spectrum}. Next we provide some numerical experiments to verify our results in previous theorems.
\begin{example}
Let $p=5$, $n=4$, $d=p^n-3=622$ and $\eta$ be the quadratic character of $\gf_{5^4}$. Then, one has $\eta(2)=\eta(-1\pm 2\sqrt{2})=1$, and $\eta(-3)=\eta(6)=1$. Thus, by Theorems \ref{w5} and \ref{w3}, we have $\omega_5=2\,\,{\rm and}\,\,\omega_3=2.$
By Theorem \ref{thm gama}, we get $\Gamma_{5,4}=14$ and $\lambda_{1,5^4}=13$. Then, by Theorem \ref{w2}, we obtain $\omega_2=152.$ By Lemma \ref{Lem-g1} we have $T_1=4$ and by Theorem \ref{L4} one gets $M=1182481$.
By Theorem \ref{spectrum}, we get $\omega_0=236,\,\,\omega_1=209\,\,{\rm and}\,\,\omega_4=24.$

The result of the above computation can also be obtained directly by Corollary \ref{p=5ds}, and it is in accordance with the differential spectrum of the mapping $x^{622}$ over $\gf_{5^4}$ calculated directly by Magma, which is
$$\left[\omega_0=236,\,\,\omega_1=209,\,\,\omega_2=152,\,\,\omega_3=2,
\,\,\omega_4=24,\,\,\omega_5=2\right].$$

\end{example}

\begin{example}
Let $p=5$, $n=5$, $d=p^n-3=3122$ and $\eta$ be the quadratic character of $\gf_{5^5}$. Then, one has $\eta(2)=\eta(-3)=-1$. Thus, by Theorems \ref{w5} and \ref{w3}, we have $\omega_5=0\,\,{\rm and}\,\,\omega_3=0.$
We get $\Gamma_{5,\,5}=82$ and $\lambda_{1,5^5}=83$ by Theorem \ref{thm gama}. Then, by Theorem \ref{w2}, we obtain $\omega_2=760.$ By Lemma \ref{Lem-g1} we have $T_1=0$ and by Theorem \ref{L4} one gets $M=29524925$.
By Theorem \ref{spectrum}, we get
$\omega_0=1180,\,\,\omega_1=1045\,\,{\rm and}\,\,\omega_4=140.$

From Corollary \ref{p=5ds}, we can get the same result directly. The above result is also in accordance with the numerical result obtained from computer experiments, which is
$$\left[\omega_0=1180,\,\,\omega_1=1045,\,\,\omega_2=760,\,\,\omega_3=0,\,\,\omega_4=140\right].$$

\end{example}

\begin{example}
Let $p=7$, $n=4$, $d=p^n-3=2398$. By Theorems \ref{w5} and \ref{w3}, we have $\omega_5=0\,\,{\rm and}\,\,\omega_3=0.$
By Theorem \ref{thm gama}  and Remark \ref{remaka=0}, we get $\lambda_{1,7^4}=-99$. Then, by Theorem \ref{w2}, we obtain $\omega_2=624.$ By Lemma \ref{Lem-g1} we have $T_1=3$ and by Theorem \ref{L4} one gets $M=17056801 $.
By Theorem \ref{spectrum}, we get
$\omega_0=888,\,\,\omega_1=801\,\,{\rm and}\,\,\omega_4=88.$

The above result can also be obtained directly by Corollary \ref{corop7}, and it coincides with the numerical result computed by Magma, which is
$$[\omega_0=888,\,\,\omega_1=801,\,\,\omega_2=624,
\,\,\omega_3=0,\,\,\omega_4=88].$$

\end{example}

\section{concluding remarks}\label{sec5}
In this paper, we determine the differential spectrum of power function $x^{p^n-3}$ over $\gf_{p^n}$ for all primes $p\geq3$ with a unified approach. It is interesting that the differential spectrum of $x^{p^n-3}$ has a close connection with the quadratic character sums $\lambda_{1,p^n}$ defined as in (\ref{2:f1}) and $\lambda_{2,p^n}$  in (\ref{2:f2}). When $p\ge 5$, these two quadratic character sums are all related to the quadratic character sum  $\Gamma_{p,\,n}$, which can be evaluated by the theory of elliptic curves over finite fields. As a result, the differential spectrum of $x^{p^n-3}$ over $\gf_{p^n}$ is completely determined in the sense that for any given odd prime $p$, all its coordinates can be expressed explicitly in terms of $n$. Our result resolves a problem that is left open for twenty years, and includes a recent result in \cite{XZLH} as a special case.

%\section{Acknowledgements}

\end{document}